\documentclass[conference,letterpaper]{IEEEtran}

\usepackage[utf8]{inputenc} 
\usepackage[T1]{fontenc}
\usepackage{url}
\usepackage{ifthen}
\usepackage{cite}
\usepackage[cmex10]{amsmath}  

\usepackage{dsfont}
\usepackage{amsthm,amssymb,amsfonts} 
\usepackage{graphicx}
\usepackage{braket}
\usepackage{dcolumn}
\usepackage{bm}
\usepackage[linesnumbered,ruled]{algorithm2e}
\usepackage{algorithmic} 
\usepackage{booktabs}

\newtheorem{theorem}{Theorem}

\newtheorem{lemma}{Lemma}
\newtheorem{claim}{Claim}

\newtheorem{remark}{Remark}
\newtheorem{definition}{Definition}

\newcommand{\calB}{\mathcal{B}}
\newcommand{\calC}{\mathcal{C}}

\newcommand{\calX}{\mathcal{X}}
\newcommand{\calY}{\mathcal{Y}}

\newcommand{\bfy}{\mathbf{y}}

\newcommand{\bfz}{\mathbf{z}}

\newcommand{\bbF}{\mathbb{F}}

\begin{document}

\title{Communication complexity of entanglement assisted multi-party computation}


\author{ 
  \IEEEauthorblockN{Ruoyu Meng and Aditya Ramamoorthy}
  \IEEEauthorblockA{Department of Electrical and Computer Engineering\\ 
                    Iowa State University, Ames, IA , USA\\
                    Email: \{rmeng, adityar\}@iastate.edu}
}


\maketitle

\begin{abstract}
We consider a quantum and a classical version of a multi-party function computation problem with \(n\) players, where players \(2, \dots, n\) need to communicate appropriate information to player 1, so that a ''generalized'' inner product function with an appropriate promise can be calculated. In the quantum version of the protocol, the players have access to entangled qudits but the communication is still classical. The communication complexity of a given protocol is the total number of classical bits that need to be communicated. When \(n\) is prime and for our chosen function, we exhibit a quantum protocol (with complexity $(n-1) \log n)$ bits) and a classical protocol (with complexity $((n-1)^2 (\log n^2)$ bits). 
We present an integer linear programming formulation for determining a lower bound on the classical communication complexity. This demonstrates that our quantum protocol is strictly better than classical protocols.
\end{abstract}

\section{\label{sec:level1}Introduction}
\label{sec:intro}
We consider a multi-party function computation scenario in this work. There are a total of $n$ players in the system numbered $1, 2, \dots, n$. Each player observes her input and players $2, \dots, n$ (remote parties) communicate an appropriate number of bits that allows player 1 to finally compute the value of the function. Clearly, this can be accomplished if players $2, \dots, n$ send their inputs but in fact in many cases, the function value can be computed with much lesser information. Thus, a natural question is to understand the minimum number of bits the remote parties need to send to player 1.

Such problems are broadly studied under the umbrella of communication complexity \cite{Yao79,Yao93} in the literature. In this work we consider the zero-error version of this problem. Our main goal is to understand the advantage that the availability of quantum entanglement confers on this problem and comparing it with classical protocols. Such problems have a long history in the literature \cite{HHHH09,BCMW10}.

{\bf Background:} Within quantum communication complexity (QCC) problems, there are three kinds of quantum protocols. In the first kind (introduced by Yao \cite{Yao93}) each player communicates via a quantum channel and the metric is the number of qubits transmitted. We call it the quantum transmission model. The second variation assumes that each player can use entanglement as a free resource but the communication is classical; the metric is the number of classical bits transmitted. We call it the entanglement model. It was introduced by Cleve and Buhrman \cite{CB97}. The third kind is a combination of the first two. We call it the combined model. It allows free usage of entanglement and works with quantum communication. The work of de Wolf \cite{W02} shows that, in the two party case, the entanglement model 
can be reduced to the quantum transmission model with a factor of two penalty using teleportation \cite{BBCJPW93}.

Buhrman, Cleve, Wigderson \cite{BCW98} and Cleve, van Dam, Nielsen and Tapp \cite{CDNT13} considered the case of the two party function computation with quantum communication and used reduction techniques to connect problems in QCC to other known problems and derived upper/lower bounds for QCC in this manner. In particular, the first work \cite{BCW98} showed examples, such as set disjointness function, where quantum protocols are strictly better than classical ones in the bounded-error setting. Here, the set-disjointness problem is such that each player has a set and wants to decide if their intersection is empty. Buhrman and de Wolf \cite{BW01} generalized the two-party "log rank" lower bound of classical communication complexity to QCC where quantum protocols use both shared entanglement and quantum communication.  
For other two-party upper/lower bound techniques, see \cite{R03,K07,VH02,MT20,LS22}. 

{\bf Related Work:}  Now we discuss works in multiparty quantum communication complexity. There are mainly two kinds of models. The number-in-hand (NIH) model assumes each player observes only one variable. The number-on-forehead (NOF) model assumes each player observes all but one variable.   Fran{\c{c}}ois and Shogo\cite{FS18} considered the NIH model with quantum communication and gave a quantum protocol for a three-party triangle-finding problem; the formulation considers bounded error. This has polynomial advantage with respect to any classical protocol. 
Here, the triangle-finding problem is such that the edge set of a graph is distributed over each user and the task is to find a triangle of the graph. 

The results in next two works hold for both NIH and NOF models.
Lee  and Schechtman  and Shraibman \cite{LSS09}  proved a Grothendieck-type inequality and then derived a general lower bound of the multiparty QCC for Boolean function in Yao's model. Following this work, Briet, Buhrman, Lee, Vidick \cite{BBLV09} showed a similar inequality for the multiparty XOR game and proved that the discrepancy method lower bounds QCC when the combined is of the third kind discussed earlier. 

Buhrman, van Dam,  H\o{}yer,  Tapp \cite{BVHT99} considered the NIH model with shared entanglement and proposed a three-party problem  with a quantum protocol that is better than any classical protocol by a constant factor. Following this work, Xue, Li, Zhang, Guo \cite{XLZG01} and Galv\~ao \cite{G01} showed  similar results under the same function with more restrictions.  
The work most closely related to our work is by Cleve and Buhrman \cite{CB97}. This paper considered the case of three players denoted Alice, Bob and Carol who have $m$-bit strings denoted $\vec{x},\vec{y}$ and $\vec{z}$ respectively. The strings are such that $\vec{x}+\vec{y}+\vec{z} = \mathds{1}$, i.e., their binary sum (modulo-2) is the all-ones vector. The goal is for Alice to compute 
\begin{align*}
    g(x,y,z) &= \sum_{i=1}^m x_i y_i z_i
\end{align*}
using binary arithmetic. We note that the communication from Bob and Carol to Alice is purely classical; however, they can use entanglement in a judicious manner.
For this particular function \cite{CB97} shows that a classical protocol (without entanglement) requires three bits of communication, whereas if the parties share $3m$ entangled qubits, then two bits of communication are sufficient.

{\bf Main Contributions:}
In this work, we consider a significant generalization of the original work of \cite{CB97}. In particular, we consider a scenario with $n$ players (for prime $n$) that observe values that lie in a higher-order finite field, with a more general promise that is satisfied by the observed values. As we consider more players and higher-order finite fields, the techniques used in the original work are not directly applicable in our setting.

Our work makes the following contributions.
\begin{itemize}
    \item We demonstrate a quantum protocol that allows for the function to be computed with  $(n-1) \log n$ bits. We use the quantum Fourier Transform as a key ingredient.
    \item On the other hand, we demonstrate a classical protocol that requires the communication of $(n-1)^2 (\log n^2)$ bits.
    \item To obtain a lower bound on the classical communication complexity, we define an appropriate integer linear programming problem that demonstrates that our quantum protocol is strictly better than any classical protocol.
\end{itemize}

This paper is organized as follows. Section \ref{sec:problem_form} discusses the problem formulation and Section \ref{sec:proposed_scheme} discusses our quantum protocol. Sections \ref{sec:classical_protocol} and \ref{sec:classical_lower_bound_sec} discuss our classical protocol and the lower bound on any classical protocol respectively.

\section{Problem formulation}
\label{sec:problem_form}
\subsection{Classical/Quantum Communication Scenarios}
Let $\calX_i, i = 1, \dots, n$ and $\calY$ denote sets in which the inputs and the output lie and $f(x_1,\dots,x_{n}):\calX_1\times\dots\times \calX_{n} \mapsto \calY$ be a multivariate function. There are $n$ players such that $i$-th player is given $x_i\in \calX_i$. The first player (henceforth, Alice) receives information from each of the players and this communication should allow her to compute $f(x_1,\dots,x_n)$. The players are not allowed to communicate with each other.

In the classical protocol, players $2$ to $n$ communicate to Alice via classical channels. In the quantum protocol, we assume that the users have shared entanglement as a free resource; however, the communication is still classical. In both scenarios the classical/quantum communication complexity is the least possible number of classical bits transmitted such that Alice can compute the function among all classical/quantum protocols.

\subsection{Generalized Inner Product Function with a Promise}
\label{sec:promise}
In this work we consider a specific multivariate function and the setting where $n \ge 3$ (number of players) is prime. Let $\mathbb{F}_n$ denote the finite field of order-$n$ and $[m] \triangleq \{1,\dots,m\}$.
The $i$-th player is given a vector $\vec{x}^i=[x^i_1 \dots x^i_m]^T \in \mathbb{F}_n^m$, i.e., each $ x^i_j \in \mathbb{F}_n$. The vectors satisfy the following ``promise'': $\forall j\in[m],$ the $j$-th component of each player's vector is s.t.
\begin{align*}
[x^1_j,\dots ,x^n_j]^T \in \{& a[1,\dots ,1]^T+ b[0,\dots,n-1]^T ~|~ a,b\in \mathbb{F}_n\},    
\end{align*}
i.e., $[x^1_j,\dots ,x^n_j]^T$ lies in a two-dimensional vector space spanned by the basis vectors $[1,\dots ,1]^T$ and $[0,1,\dots,n-1]^T$. In this case, it can be observed that $[x^1_j,\dots ,x^n_j]^T$ is either a multiple of the all-ones vector (if $b=0$) or a permutation of $[0, 1, \dots,n-1]$ (if $b \neq 0$). 
The function to be computed is the generalized inner product function given by
\begin{align}
\label{eqn:GIP}GIP(\vec{x}^1,\dots ,\vec{x}^n) = \sum_{i=1}^m \left( \prod_{j=1}^nx_i^j \right),
\end{align}
where the operations are over $\mathbb{F}_n$.

\section{Proposed Quantum Protocol}
\label{sec:proposed_scheme}
We first discuss the entangled states and unitary transforms will be used in the proposed quantum protocol in Section \ref{subsec:quantum_setting}. In Section \ref{subsec:quantum_protocol}, we discuss the quantum protocol with a proof of correctness in detail. A word about notation. In what follows for complex vectors $\vec{u},\vec{v}$, $\braket{\Vec{u},\Vec{v}} = \sum_{i} u^\dagger_i v_i$ denotes the usual inner product. On the other hand if $\vec{u},\vec{v}\in \mathbb{F}_n^m$, then $\braket{\vec{u},\vec{v}} = \sum_{i=1}^m u_i v_i$ denotes the inner product over $\mathbb{F}_n$. Moreover, $\delta_{ij}$ denotes the Kronecker delta function which equals 1 if $i=j$ and 0 otherwise. All logarithms in this paper are to the base-2.

\subsection{Entanglement Resource and Unitary Transforms Used\label{subsec:quantum_setting}}

\paragraph{Shared Entangled States.}Consider $n$ isomorphic $n$-dimensional quantum systems, where each system has a computational basis denoted $\calB = \{\ket{0},\ket{1},\dots, \ket{n-1}\}$.  There are $m$ entangled states shared among $n$ players. For $i\in[m]$, prepare the entangled state 
\begin{align}
\ket{\Phi_i}:=\frac{1}{\sqrt{n}}\sum_{k=0}^{n-1}\ket{k\dots k}. \label{eq:entangled_state_defn}    
\end{align}
The $j$-th subsystem of this entangled state is given to $j$-th player for $j = 1, \dots, n$.

\paragraph{Quantum Fourier Transform.} Let $\omega:=e^{\frac{2\pi \textrm{i}}{n}}$ denote the $n$-th root of unity. The Quantum Fourier Transform (QFT) is the following unitary map that takes
\begin{equation}\label{eqn:QFT}
\ket{j}\mapsto \frac{1}{\sqrt{n}}\sum_{k=0}^{n-1}  \omega^{jk}\ket{k}, ~\forall \ket{j} \in \calB.    
\end{equation}
Let $QFT^{\otimes l}$ denote   QFT performed over $l$ isomorphic systems.

\paragraph{Phase Shift Map.}
For $j\in \mathbb{F}_n$, we define
\begin{align}
 &P_0 \triangleq \begin{cases}
\ket{0}\mapsto \omega^{-\frac{n-1}{2}}\ket{0}\\ \nonumber
\ket{i}\mapsto\ket{i}, i\ne 0. \nonumber
\end{cases}\\
\text{If }j\ne 0,~ &P_j \triangleq 
\ket{i}\mapsto \omega^{-\frac{1}{n} (ij\text{ mod } n)}\ket{i}. \label{eq:phase_shift_map}
\end{align}

\begin{algorithm}[t]
\label{alg:quantum_protocol}
\begin{algorithmic}
\STATE{For $i\in\{1,\dots,m\}$, prepare maximally entangled ``shared state'' $\ket{\Phi_i}$ ({\it cf.} \eqref{eq:entangled_state_defn}) and distribute corresponding subsystems to all players. }
 \FOR{player $p\in\{1,\dots n\}$ } 
    \FOR{each $i\in\{1,\dots, m\}$}
    \STATE {Assume $x_i^p=j$, then player $p$ applies $P_j$ ({\it cf.} \eqref{eq:phase_shift_map}) on her part of $\ket{\Phi_i}$.} 
    \STATE {Player $p$ performs $QFT$ on her part of the shared state.} 
    \STATE {Player $p$ measures her part of the shared state in the computational basis, yielding $s_i^p\in \bbF_n$} 
    \ENDFOR
    \STATE {$s^p \leftarrow \sum_{i=1}^m s_i^p$} 
    \STATE {Player $p$ sends $s^p$ to Alice if $p\ne 1$} 
 \ENDFOR
 \STATE {$GIP(\vec{x}^1,\dots ,\vec{x}^n) =\sum_p s^p.$} 
\end{algorithmic}  
\caption{Proposed Quantum protocol}
\end{algorithm}
\subsection{The Quantum Protocol\label{subsec:quantum_protocol}}
Next, we introduce the quantum protocol.
\begin{theorem}
There exists a quantum protocol for computing $GIP(\vec{x}^1,\dots ,\vec{x}^n)$ that uses $(n-1) \log n$ bits.     
\end{theorem}
In our protocol (see Alg.~\ref{alg:quantum_protocol}), for each $i = 1, \dots, m$, each player $p$ examines $x_i^p$ and applies the corresponding phase shift map to her subsystem of $\ket{\Phi_i}$. Following this, she applies the QFT on each of her symbols and then measures in the computational basis; this yields $s_i^p \in \bbF_n$ for $i = 1,\dots, m$.  Player $p$ then transmits $\sum_{i=1}^m s_i^p$. 
As players $2 \leq p \leq n$ transmit a symbol from $\bbF_n$, it is clear that the total communication in the protocol is $(n-1)( \log n)$.

To show correctness of the protocol, we need the following auxiliary lemma. The proof appears in Appendix \ref{appendix:proof_lemma_1}.
\begin{lemma}
Let  $\vec{\alpha}=[1,\dots,1]^T\in\mathbb{F}_n^n$. Then, for each $x\in \mathbb{F}_n$, we have
    \label{lemma:qft_lemma}
\begin{align}
    QFT^{\otimes n}\left(\frac{1}{\sqrt{n}}\sum_{j=0}^{n-1} \omega^{-jx}\ket{j\cdot \vec{\alpha}}\right) = \frac{1}{n^{\frac{n}{2}}}\sum_{\vec{k}\in\{0,\dots,n-1\}^n}a_{\vec{k}}\ket{\vec{k}}.
\end{align}
Then the amplitude $a_{\vec{k}}\ne 0$ iff $\sum_{j=1}^nk_j=x$ where $\vec{k}=[k_1,\dots,k_n]^T$.  
\end{lemma}
The proof of correctness of the protocol hinges on the following lemma.
\begin{lemma}
\begin{equation}\label{relation1}
\sum_{p=1}^ns_i^p = \prod_{p=1}^nx_i^p, \text{for $i=1, \dots, m$}.
\end{equation}    
\end{lemma}
\begin{proof}
The state jointly measured by each player is
$$ QFT^{\otimes n}\left(\sum_{j=0}^{n-1} \left( \otimes_{p=1}^n P_{x_i^p} \right) \frac{1}{\sqrt{n}}\ket{j\cdot \vec{\alpha}} \right),$$
where $\vec{\alpha}=[1,\dots,1]^T$. If $[x_i^1,\dots,x_i^n]^T = [j,\dots,j]^T$, then (see Appendix \ref{lemma2} for derivation)
\begin{equation}\label{eqn:computej}
    \begin{split}
P_j^{\otimes n}
\left(\frac{1}{\sqrt{n}}\sum_{k=0}^{n-1}\ket{k\cdot 
\vec{\alpha}} 
\right) 
\mapsto 
\frac{1}{\sqrt{n}}\sum_{k=0}^{n-1}\omega^{-k{j}}\ket{k\cdot \vec{\alpha}}.
    \end{split}
\end{equation}

Thus, $QFT^{\otimes n}\left( \frac{1}{\sqrt{n}}\sum_{k=0}^{n-1}\omega^{-kj}\ket{k\cdot \vec{\alpha}} \right)$ has non-zero coefficients only for states $\ket{\vec{k}}$ such that $\sum_{l=1}^n k_l=j$ by \textbf{Lemma 1}.
Therefore, the measurement result $[s_i^1,\dots,s_i^n]^T$ must be one of $\vec{k}=[k_1,\dots,k_n]^T$ s.t. 
$$
\sum_{l=1}^nk_l=j \overset{(a)}{=} j^n=\prod_{p=1}^nx_i^p
$$
where $(a)$ follows from the fact that $j \in \mathbb{F}_n$. 

Now assume 
$[x^1_j,\dots ,x^n_j]^T = a[1,\dots ,1]^T+b[0,1,\dots,n-1]^T$ with $b\ne 0$. This implies that $a+b\cdot i\in\bbF_n$ is distinct for each $i\in\{0,\dots,n\}$. It can be observed that $a[1,\dots,1]+b[0,\dots,n-1]$ is a permutation of $[0,\dots,n-1]$, so it suffices to discuss $[x_i^1,\dots,x_i^p]^T = [0,\dots,n-1]^T$ by symmetry.
We have that (see Appendix \ref{lemma2} for derivation)
\begin{equation}\label{eqn:computeaj+b0...n-1}
    \begin{split}&P_0\otimes \dots \otimes P_{n-1}\left(\frac{1}{\sqrt{n}}\sum_{k=0}^{n-1}\ket{k\cdot \vec{\alpha}} \right) \mapsto\frac{1}{\sqrt{n}} \omega^{-\frac{n-1}{2}}\sum_{k=0}^{n-1}\ket{k\cdot \vec{\alpha}} 
    \end{split}
\end{equation}
It follows that
\begin{align*}
&QFT^{\otimes n}\left(\frac{1}{\sqrt{n}}\omega^{-\frac{n-1}{2}}\sum_{k=0}^{n-1}\ket{k\cdot \vec{\alpha}}\right)\\&=\omega^{-\frac{n-1}{2}}QFT^{\otimes n}\left(\frac{1}{\sqrt{n}}\sum_{k=0}^{n-1}\ket{k\cdot \vec{\alpha}}\right)
=\frac{1}{n^{\frac{n}{2}}}\omega^{-\frac{n-1}{2}}\sum_{\vec{k}\in\bbF_n^n}a_{\vec{k}}\ket{\vec{k}}.    
\end{align*}
By \textbf{Lemma 1}, $a_{\Vec{k}}\ne 0$ iff  $\sum_{l=1}^nk_l=0$. The measurement result $[s_i^1,\dots,s_i^n]^T$ must be $\vec{k}=[k_1,\dots,k_n]^T$ such that 
$$
\sum_{l=1}^nk_l=0=\prod_{j=0}^{n-1}j=\prod_{p=1}^nx_i^p.
$$
\end{proof}

Now, we show that our protocol computes $GIP(\Vec{x}^1,\dots,\vec{x}^n)$ correctly. Since $s^p = \sum_i s_i^p$, by applying (\ref{relation1}), we have that
\begin{align*}
\sum_ps^p &= \sum_p\sum_i s_i^p= \sum_i\sum_p s_i^p\\&= \sum_i\prod_{p=1}^nx_i^p= GIP(\Vec{x}^1,\dots,\vec{x}^n).
\end{align*}

\section{Proposed Classical protocol}
\label{sec:classical_protocol}
We now move on to considering purely classical protocols for our problem, i.e., ones that do not consider entanglement. At the top-level our classical scheme operates by communicating the ``number'' of different symbols that exist in within each player's vector. We show that this suffices for Alice to recover the function value.

More precisely, let $\beta_k^p$ be the number of ``$k$'' values in the vector of $p$-th player; recall that player $p$ is assigned $\vec{x}_p = x_1^p\dots x_m^p$. Note that $\sum_{k=0}^{n-1}\beta_k^p=m$.
\begin{theorem}
\label{thm:classical_protocol}
    There exists a classical protocol for computing $GIP(\vec{x}^1,\dots ,\vec{x}^n)$ that uses $(n-1)^2 (\log n^2)$ bits.
\end{theorem}
In our protocol (see Alg.~\ref{alg:classical_protocol}), for each $i = 1, \dots, n-1$, each player $p$ transmits $\beta_k^p \text{ (mod}\ n^2\text{)}$. Alice computes each $\beta_0^p \text{ (mod}\ n^2\text{)}$ by using the fact that $\sum_{k=0}^{n-1}\beta_k^p=m$. Finally, Alice computes the value of the function by using $\{\beta_k^p \text{ (mod}\ n^2\text{)}|k\in\bbF_n,p\in[n]\}$
For each player $p\in\{2,\dots,n\}$, $p$  transmits $\{\beta_1^p\text{ (mod}\ n^2\text{)},\dots, \beta_{n-1}^p\text{ (mod}\ n^2\text{)}\}$.   The total number of bits transmitted is $(n-1)^2(\log n^2)$. The proof of the Theorem \ref{thm:classical_protocol} appears in Appendix \ref{appendix:proof_classical_protocol}.

\begin{algorithm}[t]
\label{alg:classical_protocol}
\begin{algorithmic}
\FOR{player $p\in\{1,\dots n\}$} 
    \FOR{each $k\in\bbF_n$} 
    \STATE{$\beta_k^p\leftarrow$ number of ``$k$'' values in $x_1^p\dots x_m^p$} 
    \IF{$p$ is not Alice \textbf{and} $k\ne 0$} 
        \STATE {$p$ sends $\beta_k^p\text{ (mod}\ n^2\text{)}$ to Alice} 
    \ENDIF
    \ENDFOR
\ENDFOR
\FOR{$p\in\{2,\dots n\}$} 
    \STATE{Alice computes $\beta_0^p\text{ (mod}\ n^2\text{)}$
    by using 
    $\sum_{k=0}^{n-1}\beta_k^p=m$.}
\ENDFOR
\STATE{
$W \leftarrow \sum_{p=1}^n\sum_{k=1}^{n-1}k\cdot \beta_k^p+ \frac{n^2-n}{2}(n-1) \sum_{p=1}^n \beta_0^p \text{ (mod}\ n^2\text{)}$}
\STATE{$GIP(\vec{x}^1,\dots ,\vec{x}^n) = W / n$}
\end{algorithmic}
\caption{Proposed Classical protocol}
\end{algorithm}

\section{Classical communication complexity lower bound}
\label{sec:classical_lower_bound_sec}
We now discuss a lower bound on the classical communication complexity   that demonstrates a strict separation between our proposed quantum protocol and any classical protocol. Analytically, this seems to be a rather hard problem, and we discuss it as an item for future work. We are able to show however, the strict separation numerically using ILPs (see Section \ref{sec:classical_lower_bd} below). In addition, we present an analytical argument below that demonstrates that for $n=3$, the communication cost of {\it any} classical protocol is at least $2 \log_2 3$. 

We assume that Alice, Bob, and Carol are given vectors $\vec{x}^1$, $\vec{x}^2$, and $\vec{x}^3$, respectively, each of length $m$. The promise ({\it cf.} Sec. \ref{sec:promise}) is equivalent to 
\begin{align}
x^1_j+x^2_j+x^3_j=0,~ \text{for~}j=1,\dots,m.    \label{eq:promise_for_three_users}
\end{align}
This implies that the GIP function in this case can be computed if we know any two out of $\vec{x}^1, \vec{x}^2$, and $\vec{x}^3$. We assume that Carol labels her sequences ($\vec{x}^3 \in \{0,1,2\}^m$) with one of at most three possible labels.
 We denote this label by a mapping $\beta: \{0,1,2\}^m \to \{0,1,2\}$. Recall that Alice knows $\vec{x}^1$. 
\begin{definition}
    We define Bob's confusion graph $G_B = (V_B, E_B)$ as follows. The vertex set $V_B$ corresponds to the $3^{m}$ sequences $\vec{x}^2 \in \{0,1,2\}^m$. The $i$-th such sequence is denoted $\vec{x}^2[i]$ for $i = 0, \dots, 3^m -1$, with similar notation for Alice and Carol's sequences.
    
     There exists an edge $(\vec{x}^2[i], \vec{x}^2[j]) \in E_B$, for $i\neq j$ if there exists an Alice sequence $\vec{x}^1[*]$ and Carol sequences $\vec{x}^3[a]$ and $\vec{x}^3[b]$ such that (i) $\beta(\vec{x}^3[a]) = \beta(\vec{x}^3[b])$ (note that we allow $a=b$), and (ii) $GIP(\vec{x}^1[*],\vec{x}^2[i],\vec{x}^3[a]) \neq  GIP(\vec{x}^1[*],\vec{x}^2[j],\vec{x}^3[b])$. 
\end{definition}
Note that if $(\vec{x}^2[i], \vec{x}^2[j]) \in E_B$, the Bob has to assign different labels to $\vec{x}^2[i]$ and $\vec{x}^2[j]$; otherwise, Alice has no way to compute the function with zero error. The idea of the confusion graph goes back to the work of Shannon \cite{S56}.

The main idea of the argument below is to show that there exists a triangle in $G_B$. This implies that the Bob needs to use at least three labels for Alice to decode with zero-error.

Since Carol uses at most three labels, then using the pigeon-hole principle there exist at least $3^{m-1}$ sequences that have the same Carol label. Let us denote this set $\calC$.
\begin{claim}
       There is a subset of two coordinates where all nine patterns $\{0,1,2\}^2$ appear within the sequences in $\calC$.
\end{claim}
\begin{proof}
    Suppose that $m$ is even. Then, we can partition the coordinates as $\{1,2\},\{3,4\}, \dots, \{m-1,m\}$. Let us arrange the sequences in $\calC$ as rows; the number of rows is $|\calC| \geq 3^{m-1}$. Now suppose that the projection onto any pair of coordinates has at most 8 representatives. Then, the size of $\calC$ can be at most $8^{m/2}$. For large enough $m$, we have
    \begin{align*}
        \frac{3^{m-1}}{8^{m/2}} &= \frac{1}{3} \times \left(\frac{9}{8}\right)^{m/2} > 1.
    \end{align*}
\end{proof}

Without loss of generality, we assume that all nine patterns occur within the first two coordinates of $\calC$. We pick 9 such representatives from $\calC$ and denote them $\bfz_{00}, \bfz_{01}, \bfz_{02}, \bfz_{10}, \dots, \bfz_{22}$; the subscript corresponds to the values on the first two coordinates.

Let us pick Alice's sequence $\vec{x}^1[*] = [1 ~ 1~0~\dots~0]$. Corresponding to this $\vec{x}^1[*]$, for the Carol sequences $\bfz_{00}, \dots, \bfz_{22}$, using the promise we can obtain the corresponding Bob sequences $\bfy_{00}, \dots, \bfy_{22}$. We note here that 
\begin{align*}
    \bfy_{00} &= [2~2 -z_{00}(3:m)], \text{~and}\\
    \bfy_{01} &= [2~1 -z_{01}(3:m)], \text{~and}\\
    \bfy_{11} &= [1~1 -z_{11}(3:m)].
\end{align*}
where, e.g., $z_{00}(3:m)$ denotes the components of vector $z_{00}$ from index 3 onwards (basically MATLAB notation). 
\begin{claim}
    In Bob's confusion graph, $G_B$, the sequences $\bfy_{00},\bfy_{01},$ and $\bfy_{11}$ form a triangle.
\end{claim}
\begin{proof}
    We need to examine $GIP(\vec{x}^1[*],\bfy_i,\bfz_i)$ for $i=00,01,11$. Only the first two coordinates matter since $\vec{x}^1[*] = [1 ~ 1~0~\dots~0]$. The corresponding evaluations are $0,1,2$ which are pairwise different.
\end{proof}
The above 
shows that Bob needs to use at least three labels for Alice to decode with zero-error. By symmetry, Carol needs to use three labels as well. To summarize, the communication complexity of any classical protocol is at least $2 \log 3$ bits.

 \begin{remark}
     It may be possible to use a variant of the above combinatorial argument to establish that the chromatic number of $G_B$ is strictly larger than three. However, this does not seem to follow in a straightforward manner.
 \end{remark}

\begin{table}[t]
\caption{Numerical results.}
\label{tab:table1}
\centering
\begin{tabular}{lccc}
\toprule
\textrm{$m$} & \textrm{$l^b$} & \textrm{$l^c$} & \textrm{Feasibility}\\
\midrule
1 & 1 & 3 & Feasible\\
3 & 1 & 17 & Infeasible\\
2 & 2 & 4 & Infeasible\\
3 & 3 & 3 & Infeasible\\
2 & 3 & 4 & Feasible\\
3 & 5 & 5 & Feasible\\
\bottomrule
\end{tabular}
\end{table}

\subsection{ILP Feasibility Problem for Classical Lower Bound}
\label{sec:classical_lower_bd}
We now present a lower bound on the communication complexity of any classical protocol for our problem. Towards this end we pose this as an integer linear programming problem that can be solved numerically.

Suppose, for $p\in\{2,\dots,n\}$, the $p$-th player  sends symbols (labels) in $[l^p] := \{1,2,\dots ,l^p\}$ for some large enough positive integer $l^p$. 
Let $c \in[l^p]$ and define $I_{\vec{x}^p, c}\in\{0,1\}$ to be the indicator that the $p$-th player sends
 $c$  when it has the vector $\vec{x}^p\in \bbF_n^m$. As this mapping is unique, we have $\sum_{c\in[l^p]} I_{\Vec{x}^p,c} = 1$.  Furthermore, for a given set of vectors $\vec{x}^p$ for $p\in\{2,\dots,n\}$ if the $p$-th player sends label $c^p$, we have $\prod_{p=2}^nI_{\vec{x}^p, c^p} = 1$.
 
Consider two sets of vectors $\{\vec{x}^p\in\bbF_n^m|p\in\{1,\dots,n\}\}$, $\{\vec{z}^p\in\bbF_n^m|p\in\{1,\dots,n\}\}$. We denote $(\vec{x}^1,\dots,\vec{x}^n) \sim_{GIP} (\vec{z}^1,\dots,\vec{z}^n)$
if the following conditions are satisfied.
\begin{enumerate}
    \item Both $(\vec{x}^1,\dots,\vec{x}^n)$ and $(\vec{z}^1,\dots,\vec{z}^n)$ satisfy the promise ({\it cf.} Sec. \ref{sec:promise}).
    \item $\vec{x}^1 = \Vec{z}^1$.
    \item $GIP(\vec{x}^1,\dots, \vec{x}^n)\ne GIP(\vec{z}^1,\dots, \vec{z}^n)$.
\end{enumerate}
This definition  applies to distinct inputs with the "same" Alice vector, but different function evaluations. It can be seen that for two such distinct inputs, the symbols communicated by players 2 to $n$ have to be distinct, otherwise Alice has no way to decode in a zero-error fashion.

Our proposed ILP works with fixed $l^p$'s and a fixed value of $m$. Owing to complexity reasons $m$ cannot be very large. However, we point out that if the ILP is infeasible for given $l^p$'s and a $\tilde{m}$, then our lower bound holds for arbitrary values $m \geq \tilde{m}$. To see this we note that our lower bound would continue to hold even if Alice was provided the values $x^p_{\tilde{m} + 1}, \dots, x^p_{m}$ for all players $p = 2, \dots, n$.

Consider a $0-1$ integer programming feasibility problem:
\begin{equation}\label{eq:optimization}
\begin{aligned}
&\min \quad  0\\
\textrm{s.t.} \quad &p\in\{2,\dots,n\},c\in[l^p],\Vec{x}^p\in\bbF_n^m,\\
&I_{\Vec{x}^p,c} \in \{0,1\},\\
&\sum_{c\in[l^p]} I_{\Vec{x}^p,c} =  1, ~\forall~ \vec{x}^p,\\\
\sum_{c^2\in[l^2],\dots,c^n\in[l^n]} &|\prod_{p=2}^nI_{\vec{x}^p, c^p} - \prod_{p=2}^nI_{\vec{z}^p, c^p}| =2\\
\quad \text{ for all }&(\vec{x}^1,\dots,\vec{x}^n) \sim_{GIP}(\vec{z}^1,\dots,\vec{z}^n).
\end{aligned}
\end{equation}
The infeasibility of the above  problem corresponds to a lower bound on the  classical communication complexity.  The proof of the following theorem appears in Appendix \ref{appdx:thm3}.
\begin{theorem}\label{thm:ILP}
There exists a deterministic classical protocol computing $GIP(\cdot)$ where each player sends at most $l^p$ different labels for $p\in\{2,\dots,n\}$ iff the above integer programming is feasible. 
\end{theorem} 
\begin{remark}
    The integer program contains constraints that involve the product of variables and equality constraints with sums of absolute values. We show how these constraints can be linearized in Appendix \ref{appendix:linearizing_constr}. The entire code for our ILP is available at this online repository \cite{github_code}.
\end{remark}
\subsection{Numerical experiments} 
In our numerical experiments, we considered an instance of the ILP involving $n=3$ players, namely Alice, Bob, and Carol. Let $m$ represent the length of each vector, while $[l^b],[l^c]$ denote the sets of labels used by Bob and Carol, with $l^b,l^c$ as the sizes of these sets. 

We assume that Alice, Bob, and Carol are given vectors $\vec{x}^1$, $\vec{x}^2$, and $\vec{x}^3$, respectively, each of length $m$. In this case, the promise is given by \eqref{eq:promise_for_three_users}.
It can be observed that swapping the vectors of Bob and Carol continues to satisfy the promise. Due to this inherent symmetry, a protocol with communication lengths $l^b=x$ and $l^c=y$ exhibits the same feasibility as one with $l^b=y$ and $l^c=x$. Consequently, for the ILP we can assume that $l^b \le l^c$.

The experimental results under varying settings of $l^b,l^c,m$ are displayed in TABLE \ref{tab:table1}. It shows for instance, that when $l^b = 1$ and $l^c = 17$, the ILP is infeasible with $m=3$. This implies that for a feasible classical protocol, with $l^b = 1$, we need at least $\log (18)$ bits to be transmitted from Carol. Similarly, the triplets $(m,l^b,l^c) = (2,2,4)$ 
and $(3,3,3)$ are infeasible. This implies that when $l^b$ equals 2 or 3 the sum rate is $ \geq \min(\log 2 + \log 5, \log 3 + \log 4)$. 


Recalling that our proposed protocol employs $2\log(3)$ bits of communication, and by the fact that 
\begin{align*}
  2 \log 3 < \min(\log 18, \log 3 + \log 4, \log 2 + \log 5)  
\end{align*}
we conclude that there is a strict separation between our quantum protocol and any classical protocol.



\newpage

\bibliographystyle{IEEEtran}
\bibliography{ref}

\newpage

\appendices

\section{Proof of Lemma \ref{lemma:qft_lemma}
\label{appendix:proof_lemma_1}}
Recall that $$\vec{\alpha}=[1,\dots,1]^T.$$
The action of $\text{QFT}^{\otimes n}$ on  $\frac{1}{\sqrt{n}}\sum_{j=0}^{n-1} \omega^{-xj}\ket{j\cdot \vec{\alpha}}$ is 
    \begin{align*}
 \frac{1}{n^{\frac{n}{2}}}\sum_{\vec{k}\in\bbF_n^n}&a_{\vec{k}}\ket{\vec{k}}=QFT^{\otimes n}\left(\frac{1}{\sqrt{n}}\sum_{j=0}^{n-1} \omega^{-xj}\ket{j\cdot \vec{\alpha}}\right)
\\& = 
 \frac{1}{n^{\frac{n}{2}}}\sum_{\vec{k}\in\bbF_n^n} \left(\frac{1}{\sqrt{n}}\sum_{j=0}^{n-1} \omega^{-xj}\cdot \omega^{\langle \vec{k},j\cdot \vec{\alpha} \rangle}\right)\ket{\vec{k}}.
\end{align*}
Write $\Vec{k}=[k_1,\dots,k_n]^T$. Therefore, $$a_{\vec{k}}=\frac{1}{\sqrt{n}}\sum_{j=0}^{n-1} \omega^{-xj}\cdot \omega^{\langle \vec{k},j\cdot \vec{\alpha}\rangle} = \frac{1}{\sqrt{n}}\sum_{j=0}^{n-1} \omega^{-xj+ j \sum_{l=1}^nk_l}.$$
When $\vec{k}$ satisfies $\sum_{l=1}^nk_l=x$, we have that $a_{\vec{k}} = \sqrt{n}\ne 0$. Otherwise, since $-x+\sum_{l=1}^nk_l\ne 0$ and $n$ being prime, $1-\omega^{n(-x+\sum_{l=1}^nk_l)}=0$ and $1-\omega^{-x+\sum_{l=1}^nk_l}\ne0$. We have
\begin{align*}
&\frac{1}{\sqrt{n}}\sum_{j=0}^{n-1} \omega^{-xj+ j \sum_{l=1}^nk_l}=\frac{1}{\sqrt{n}}\sum_{j=0}^{n-1}\omega^{(-x+\sum_{l=1}^nk_l) j} 
\\& = \frac{1}{\sqrt{n}}\frac{1-\omega^{n(-x+\sum_{l=1}^nk_l)}}{1-\omega^{-x+\sum_{l=1}^nk_l}} =0.
\end{align*}

\section{Derivation of equation (\ref{eqn:computej}) and \label{lemma2} equation (\ref{eqn:computeaj+b0...n-1})}
We derive equation (\ref{eqn:computej}) by considering two cases. The first case is that $[x_i^1,\dots,x_i^n]^T = [0,\dots,0]^T$. Then, we have
\begin{equation}\label{compute0}
    \begin{split}
        &P_0^{\otimes n}\left(\frac{1}{\sqrt{n}}\sum_{k=0}^{n-1}\ket{k\cdot \vec{\alpha}} \right) \mapsto \frac{1}{\sqrt{n}}\omega^{-\frac{n(n-1)}{2}}\ket{0\cdot \vec{\alpha}}\\&+\frac{1}{\sqrt{n}}\sum_{k=1}^{n-1}\ket{k\cdot \vec{\alpha}}=\frac{1}{\sqrt{n}}\sum_{k=0}^{n-1}\ket{k\cdot \vec{\alpha}}=\frac{1}{\sqrt{n}}\sum_{k=0}^{n-1}\omega^{-k \cdot 0}\ket{k\cdot \vec{\alpha}}. 
    \end{split}
\end{equation}
 The second case is that  $[x_i^1,\dots,x_i^n]^T = [j,\dots,j]^T$ with $j\ne 0$. Then, we have
\begin{equation}\label{computej}
    \begin{split}
&P_j^{\otimes n} \left(\frac{1}{\sqrt{n}}\sum_{k=0}^{n-1}\ket{k\cdot \vec{\alpha}} \right) \mapsto \frac{1}{\sqrt{n}}\sum_{k=0}^{n-1}\omega^{-n(\frac{1}{n} (kj\text{ mod }n))}\ket{k\cdot \vec{\alpha}}\\& =
\frac{1}{\sqrt{n}}\sum_{k=0}^{n-1}\omega^{-k{j}}\ket{k\cdot \vec{\alpha}}
    \end{split}
\end{equation}
where the last equality holds since $\omega^{-kj} = \omega^{- (kj \text{~mod~} n)}$. Thus, in this case collectively, we can express the joint state after phase-shifting as $\frac{1}{\sqrt{n}}\sum_{k=0}^{n-1}\omega^{-kj}\ket{k\cdot \vec{\alpha}}$.

Now we derive equation (\ref{eqn:computeaj+b0...n-1}). We have 
\begin{align*}
&P_0\otimes \dots \otimes P_{n-1} \left(\frac{1}{\sqrt{n}}\sum_{k=0}^{n-1}\ket{k\cdot \vec{\alpha}} \right) \mapsto\\& 
\frac{1}{\sqrt{n}}\omega^{-\frac{n-1}{2}} \ket{0\cdot \vec{\alpha}}+ \frac{1}{\sqrt{n}}\sum_{k=1}^{n-1}\omega^{-\sum_{j=0}^{n-1} \frac{1}{n}(kj\text{ mod }n)}\ket{k\cdot \vec{\alpha}}\\
&\overset{(a)}{=} 
\frac{1}{\sqrt{n}}\omega^{-\frac{n-1}{2}} \ket{0\cdot \vec{\alpha}}+ \frac{1}{\sqrt{n}}\sum_{k=1}^{n-1}\omega^{- \frac{1}{n}(\sum_{j=0}^{n-1}j)}\ket{k\cdot \vec{\alpha}}\\
& = 
\frac{1}{\sqrt{n}}\omega^{-\frac{n-1}{2}} \ket{0\cdot \vec{\alpha}}+ \frac{1}{\sqrt{n}}\sum_{k=1}^{n-1}\omega^{- \frac{1}{n} [\frac{n(n-1)}{2} ]}\ket{k\cdot \vec{\alpha}}
\\&=\frac{1}{\sqrt{n}} \omega^{-\frac{n-1}{2}}\sum_{k=0}^{n-1}\ket{k\cdot \vec{\alpha}} 
\end{align*}

where $(a)$ follows from the fact that  $\{kj\text{ mod }n| j\in \{0,\dots,n-1\}\}=\{0,\dots,n-1\}$ for $k\ne 0$. It follows that
\begin{align*}
&QFT^{\otimes n}\left(\frac{1}{\sqrt{n}}\omega^{-\frac{n-1}{2}}\sum_{k=0}^{n-1}\ket{k\cdot \vec{\alpha}}\right)\\&=\frac{1}{\sqrt{n}}\omega^{-\frac{n-1}{2}}QFT^{\otimes n}\left(\sum_{k=0}^{n-1}\ket{k\cdot \vec{\alpha}}\right)\\&
=\frac{1}{\sqrt{n}}\omega^{-\frac{n-1}{2}}\sum_{\vec{k}\in\bbF_n^n}a_{\vec{k}}\ket{\vec{k}}.    
\end{align*}

\section{Proof of Theorem \ref{thm:classical_protocol}}
\label{appendix:proof_classical_protocol}
For $a,b\in\bbF_n$, we define $M_{a,b} = \{i\in[m]|[x^1_i,\dots ,x^n_i]^T =a[1,\dots ,1]^T+b[0,1,\dots,n-1]^T\}$ and $m_{a,b} = |M_{a,b}|$. Since the promise ensures that, for each $i\in\{1,\dots,m\}$, there exists $a,b\in\bbF_n$ s.t.
$[x^1_i,\dots ,x^n_i]^T =a[1,\dots ,1]^T+b[0,1,\dots,n-1]^T\in \bbF_n^n$, we have that $\{M_{a,b}|a,b\in\bbF_n\}$ forms a partition of the set $\{1,\dots,m\}.$

When $i\in M_{a,0}$, i.e. $[x^1_i,\dots ,x^n_i]^T = a  [1,\dots,1]^T$, we have 
\begin{equation}
    \label{compute3}
    \prod_{p=1}^nx_i^p = a^n = a.
\end{equation} 
Otherwise, we  have $i\in M_{a,b}$ for some $b\ne 0$, so $[x^1_i,\dots ,x^n_i]^T = a [1,\dots,1 ]^T+b[0,\dots,n-1]^T$. Then, 
\begin{equation}
    \label{compute4}
    \prod_{p=1}^nx_i^p = \prod_{i=0}^{n-1}(a+i\cdot b) =0.
\end{equation}
By (\ref{compute3}) and (\ref{compute4}), if $i\in M_{a,b}$, then $\prod_{p=1}^nx_i^p = \delta_{0b}\cdot a.$ Define 
$$
\mathds{1}((x^1_i,\dots ,x^n_i), M_{a,b}) = \begin{cases}
1, & i\in M_{a,b}\\
0, & \text{otherwise},
\end{cases}
$$
i.e., it is the indicator of $i\in M_{a,b}$. 
Since $i\in M_{a,b}$ for exactly one choice of $(a,b)$, it follows that
\begin{equation}
    \label{compute5}
    \prod_{p=1}^nx_i^p = \delta_{0b}\cdot a=\sum_{a,b=0}^{n-1} \delta_{0b}\cdot a \cdot  \mathds{1}((x^1_i,\dots ,x^n_i), M_{a,b}).
\end{equation}
We have
\begin{equation}\label{classical_protocol}
\begin{split}
&\sum_{i=1}^m\prod_{p=1}^nx_i^p\\
\overset{\text{(*)}}{=}&\sum_{i=1}^m\sum_{a,b=0}^{n-1}a \cdot \delta_{0b}\cdot  \mathds{1}((x^1_i,\dots ,x^n_i), M_{a,b})\\ 
=&\sum_{a,b=0}^{n-1}\sum_{i=1}^m
a \cdot \delta_{0b} \cdot \mathds{1}((x^1_i,\dots ,x^n_i), M_{a,b})\\ 
=& \sum_{a,b=0}^{n-1}a\cdot  \delta_{0b} \cdot m_{a,b}
=  \sum_{a=0}^{n-1}a\cdot  m_{a,0}
\text{ (mod}\ n\text{)}.
\end{split}
\end{equation}
Here, (*) follows from \eqref{compute5}. Our next step is to show $\sum_{a=0}^{n-1}a m_{a,0} = W/n
\text{ (mod}\ n\text{)}$; $W$ is defined in Alg. \ref{alg:classical_protocol}.

Suppose $i\in M_{a,b}$, then $(x^1_i,\dots ,x^n_i) = a (1,\dots,1) +b(0,\dots,n-1)$. For the $p-$th player, $a+ (p - 1) b$ is the value of the $i-$th coordinate of her vector $\vec{x}^p$. For a fixed $k\in \bbF_n$, the set of $(a,b)$ s.t. $a + (p - 1) b = k$ is $\{(k,0),(k+p-1,-1),(k+2p-2,-2),\dots,(k+(n-1)(p-1),-(n-1)\}$. It follows that $\forall p\in\{1,\dots,n\},k\in\bbF_n,$

\begin{equation}\label{relation4}
\beta_k^p = \sum_{i=0}^{n-1}m_{k + i(p-1), -i}.    
\end{equation}
Consider $\sum_{i=0}^{n-1}\sum_{p=1}^{n}m_{k + i(p-1), -i}$. When $i=0$, $m_{k+i(p-1),-i}$ is counted $n$ times. Denote $S = \{(0,0),\dots,(n-1,0)\}.$ For arbitrary $
[x,y]^T
\in\bbF_n^2-S$, the equation 
\begin{align*}
\begin{bmatrix}
k+i(p-1)\\-i
\end{bmatrix}
&=
\begin{bmatrix}
x\\
y
\end{bmatrix}
\end{align*}
has a unique solution  given by 
\begin{align*}
i &= -y, \text{~and}\\
p &= \frac{y-x+k}{y}.
\end{align*}
Thus, we have
$$
\bbF_n^2-S=\{(k+i(p-1),-i~|~p\in\{1,\dots,n\},i\in\{1,\dots,n-1\}\}.    
$$
 Therefore, when $i\ne 0$, $m_{k+i(p-1),-i}$ is counted exactly once. 
It follows that
\begin{equation}\label{relation2}
\begin{split}
   & \sum_{p=1}^{n}\sum_{i=0}^{n-1}m_{k + i(p-1), -i}\\
= & \sum_{p=1}^{n}\sum_{i=1}^{n-1}m_{k + i(p-1), -i}  +\sum_{p=1}^{n}m_{k + 0(p-1), 0}\\
= & \sum_{a,b\in \bbF_n^2-S}m_{a,b}  +n\cdot m_{k, 0}.
\end{split}
\end{equation}Now we have
\begin{align*}
\sum_{p=1}^n\sum_{k=1}^{n-1}k\cdot \beta_k^p 
\overset{\text{\eqref{relation4}}}{=}& \sum_{k=1}^{n-1}k\sum_{p=1}^n\sum_{i=0}^{n-1}m_{k + i\cdot (p-1), -i}\\
\overset{\text{\eqref{relation2}} }{=} &
\sum_{k=1}^{n-1}k\Big[ \sum_{a,b\in \bbF_n^2-S}m_{a,b}  +n\cdot m_{k, 0}\Big]\\
= &
\frac{n^2-n}{2}\sum_{a,b\in \bbF_n^2-S}m_{a,b}+n\sum_{k=1}^{n-1}k\cdot m_{k,0}\\
\text{and }
\sum_{p=1}^n \beta_0^p 
\overset{\text{\eqref{relation4}}}{=} &\sum_{p=1}^n\sum_{i=0}^{n-1}m_{i\cdot (p-1), -i}
\overset{\text{\eqref{relation2}}}{=}\sum_{a,b\in \bbF_n^2-S}m_{a,b}  +n\cdot m_{0, 0}.
\end{align*}

Note $n>2$ is prime, so $n-1$ is divisible by 2. It follows that 
\begin{align*}
W=&\sum_{p=1}^n\sum_{k=1}^{n-1}k\cdot \beta_k^p + \frac{n^2-n}{2}\cdot(n-1) \sum_{p=1}^n \beta_0^p \\
=&\frac{n^2-n}{2}\sum_{a,b\in \bbF_n^2-S}m_{a,b}+n\sum_{k=1}^{n-1}k\cdot m_{k,0} 
\\&+  \frac{n^2-n}{2}\cdot(n-1) \Big [\sum_{a,b\in \bbF_n^2-S}m_{a,b}  +n\cdot m_{0, 0}\Big ]\\
=& n\sum_{k=1}^{n-1}k\cdot m_{k , 0} + n^2\cdot \frac{(n-1)}{2} \sum_{a,b\in \bbF_n^2-S} m_{a,b}\\
&+n^2\cdot\frac{(n-1)^2}{2}\cdot m_{0, 0}\\
=&n\sum_{k=0}^{n-1}k\cdot  m_{k , 0}\text{ (mod}\ n^2\text{)}.
\end{align*}

Divide both sides by $n$ and we get $W/n=\sum_{k=1}^{n-1}km_{k , 0}\text{ (mod}\ n\text{)}$. Now we are done because of (\ref{classical_protocol}).
\section{Proof of theorem \ref{thm:ILP}}\label{appdx:thm3}
Suppose that we have a protocol that computes the function with zero-error. Our protocol is deterministic, so, for each $\Vec{x}^p\in \bbF_n^m$, it associates exactly one label $\hat{c}$ such that the $p$-th player sends $\hat{c}$ if his vector is $\Vec{x}^p$. We set $I_{\Vec{x}^p,\hat{c}}=1$ and $I_{\Vec{x}^p,c}=0$ for all $c\ne \hat{c}.$ Therefore, $I_{\Vec{x}^p,c}\in \{0,1\}$ and $\sum_{c\in[l]} I_{\Vec{x}^p,c} = 1$ are satisfied for all choice of $\Vec{x}^p,c$.

Next, suppose we have that $(\vec{x}^1,\dots,\vec{x}^n)\sim_{GIP}(\vec{z}^1,\dots,\vec{z}^n)$. Furthermore, assume that the $p$-th player sends $s^p/t^p$ for $\vec{x}^p/\vec{z}^p$ for $p=2, \dots, n$. This implies that $\prod_{p=2}^nI_{\vec{x}^p, s^p} = 1$ and that $\prod_{p=2}^nI_{\vec{z}^p, t^p} = 1$. In addition, note that $\exists ~p$ such that $s^p \neq t^p$ as for these sequences, the symbols transmitted from users $2, \dots, n$ have to be distinct.

Thus we have 
\begin{align*}
1 = |\prod_{p=2}^nI_{\vec{x}^p, s^p}-\prod_{p=2}^nI_{\vec{z}^p, s^p}| = |\prod_{p=2}^nI_{\vec{x}^p, t^p}-\prod_{p=2}^nI_{\vec{z}^p, t^p}|
\end{align*}
and consequently
\begin{align*}
&\sum_{c_2\in[l^2],\dots,c_n\in[l^n]}|\prod_{p=2}^nI_{\vec{x}^p, c^p} - \prod_{p=2}^nI_{\vec{z}^p, c^p}|=\\
&|\prod_{p=2}^nI_{\vec{x}^p, \hat{s}^p} - \prod_{p=2}^nI_{\vec{z}^p, \hat{s}^p}|+|\prod_{p=2}^nI_{\vec{x}^p, \hat{t}^p} - \prod_{p=2}^nI_{\vec{z}^p, \hat{t}^p}|=2.  
\end{align*}
Therefore, the third constraint is satisfied.

Conversely, if we have 
 $\{I_{\Vec{x}^p,c}|\Vec{x}^p\in\bbF_n^m,c\in [l^p]\}$ that satisfy the constraints of the ILP, then we construct a classical protocol as follows. Suppose $p$-th player has vector $\Vec{x^p}$ for $p\in\{1,\dots,n\}$. Since there exists exactly one $s^p\in[l^p]$ s.t. $I_{x^p,s^p}=1$, then $p$ sends $s^p$ to Alice for $p\in\{2,\dots,n\}$. When Alice receives the symbols $s^p, p = 2, \dots, n$ from the other players she picks arbitrary $\{\vec{y}^p\in \bbF_n^m\}_{i=2}^n$ s.t. $(\Vec{x}^1,\vec{y}^2,\dots ,\vec{y}^n)$ satisfies the promise and $\forall p\in\{2,\dots,n\},I_{\vec{y}^p, s^p} = 1$.
 Then she outputs $f(\Vec{x}^1,\vec{y}^2,\dots ,\vec{y}^n)$. In what follows we show that 
$$
GIP(\Vec{x}^1,\vec{y}^2,\dots ,\vec{y}^n) = GIP(\Vec{x}^1,\vec{x}^2,\dots ,\vec{x}^n).
 $$
To see this assume otherwise. Then, we have
$GIP(\Vec{x}^1,\vec{y}^2,\dots ,\vec{y}^n) \ne GIP(\Vec{x}^1,\vec{x}^2,\dots ,\vec{x}^n).$ It follows that
$$
(\Vec{x}^1,\vec{x}^2,\dots ,\vec{x}^n)\sim_{GIP}(\Vec{x}^1,\vec{y}^2,\dots ,\vec{y}^n).
$$
Owing to the third constraint, we have that 
$$
\sum_{c^2,\dots,c^n\in[l]}|\prod_{p=2}^nI_{\vec{x}^p, c^p} - \prod_{p=2}^nI_{\vec{y}^p, c^p}|=2.
$$
However, we have that   $I_{\vec{x}^i, s^i}=I_{\vec{y}^i, s^i} = 1$ for all $i\in\{2,\dots,n\}$. By the first and second constraint, we have that $I_{\vec{x}^i, c^i}=I_{\vec{y}^i, c^i} = 0$ for all $i\in\{2,\dots,n\}$ and $c^i\ne s^i$. Therefore, $$
\sum_{c^2,\dots,c^n\in[l]}|\prod_{p=2}^nI_{\vec{x}^p, c^p} - \prod_{p=2}^nI_{\vec{y}^p, c^p}|=0.
$$
This gives the desired contradiction.

\section{Linearizing constraints in the integer programming problem}
\label{appendix:linearizing_constr}
One issue with the optimization problem in \eqref{eq:optimization} is that the third constraint has multiple absolute value and products of variables. Here we transform the constraints and add extra variable in \eqref{eq:optimization} to get the desired ILP. 

Our first step is to introduce auxiliary 0-1 variables that correspond to the product of other 0-1 variables. For instance, it can be verified that we can handle $\prod_{i=1}^k x_i = x'$ as follows.
\begin{align}
   x' &\leq x_i, \text{~for~} i = 1, \dots, k \\
   x' &\geq \sum_{i=1}^k x_i - (k-1).
\end{align}

As a first step we introduce such auxiliary variable for all terms that involve products of our indicator function in \eqref{eq:optimization}.

Following this step, we are left with handling constraints that involve sums of absolute values of differences. For this step we show how to replace each absolute value difference by another auxiliary variable. In particular, we can replace $|x-y|$ by $z$ as follows.

\begin{align*}
|x-y| = |x-y|^2
= x^2 + y^2 - 2xy
=x + y - 2xy
\end{align*}
where the last step follows from the fact that the variables are of type $0-1$. The product term $2xy$ can be linearized as described previously. Following these steps, all constraints in the integer programming problem are linear.

\end{document}